\documentclass[a4paper]{article}
\pdfoutput=1

\usepackage{amsmath, amssymb, amsthm}
\usepackage{fullpage}
\usepackage{enumerate}
\usepackage{graphicx}
\usepackage{authblk}
\usepackage[hidelinks]{hyperref}

\newtheorem{theorem}{Theorem}
\newtheorem{Lemma}{Lemma}

\newtheorem{Conjecture}{Conjecture}

\usepackage[ruled]{algorithm2e}

\def\S{\mathcal{S}}
\def\P{\mathcal{P}}

\def\ie{{i.e.}}
\def\star{\textasteriskcentered}
\def\dstar{\textasteriskcentered\textasteriskcentered}
\renewcommand{\emptyset}{\varnothing}
\newcommand{\kmer}[1]{\boldsymbol{#1}}
\newcommand{\edit}[2]{\textrm{edit}\left(\kmer{#1}, \kmer{#2}\right)}

\def\marrow{{\marginpar[\hfill$\longrightarrow$]{$\longleftarrow$}}}
\def\ke#1{\textcolor{red}{{\sc Ke says: }{\marrow\sf #1}}}
\def\mingfu#1{\textcolor{blue}{{\sc Mingfu says: }{\marrow\sf #1}}}
\newcommand{\commentblock}[1]{}

\title{Locality-sensitive bucketing functions for the edit distance \footnote{This work is supported by the US National Science Foundation (DBI-2019797 to M.S.) and the US National Institutes of Health (R01HG011065 to M.S.).}}
\renewcommand\footnotemark{}
\date{}

\author[1]{Ke Chen}
\author[1,2]{Mingfu Shao}
\affil[1]{\footnotesize Department of Computer Science and Engineering, School of Electronic Engineering and Computer Science, The Pennsylvania State University, United States}
\affil[2]{\footnotesize Huck Institutes of the Life Sciences, The Pennsylvania State University, United States}

\begin{document}
\maketitle

\begin{abstract}
  Many bioinformatics applications involve bucketing a set of sequences 
  where each sequence is allowed to be assigned into multiple buckets.
  To achieve both high sensitivity and precision, bucketing methods are desired to
  assign similar sequences into the same bucket
  while assigning dissimilar sequences into distinct buckets.
  Existing $k$-mer-based bucketing methods have been efficient
  in processing sequencing data with low error rate, but encounter much reduced
  sensitivity on data with high error rate.
  Locality-sensitive hashing~(LSH) schemes are able to mitigate this issue through tolerating
  the edits in similar sequences, but state-of-the-art methods still have large gaps.
  Here we generalize the LSH function by allowing it to hash one sequence into multiple buckets.
  Formally, a bucketing function, which maps a sequence (of fixed length) into a
  subset of buckets, is defined to be $(d_1, d_2)$-sensitive if any two sequences
  within an edit distance of $d_1$ are mapped into at least one shared bucket,
  and any two sequences with distance at least $d_2$ are mapped into disjoint subsets of buckets.
  We construct locality-sensitive bucketing~(LSB) functions with a variety of values of $(d_1,d_2)$
  and analyze their efficiency with respect to the total number of buckets
  needed as well as the number of buckets that a specific sequence is mapped to.
  We also prove lower bounds of these two parameters in different settings
  and show that some of our constructed LSB functions are optimal.
  These results provide theoretical foundations for their practical use
  in analyzing sequences with high error rate while also 
  providing insights for the hardness of designing ungapped LSH functions.
\end{abstract}

\section{Introduction}
\commentblock{
=======
binning sequences to avoid all-vs-all comparison:
 - UHS
 - overlap detection
 - seed/extend / dBG uses this idea / metagenomics binning (?)

LSH + or-amplification
 - LSH related papers 
 - one kmer can be assigned to multi-bins
 - embedding is hard

our definition

generalize (deterministic) LSH: when |f(s)| = 1
more powerful when allowing |f(s)| > 1 (? or move to Discussion)

summarize our results
}

Comparing a set of given sequences is a common task involved in
many bioinformatics applications, such as 
homology detection~\cite{chen2018comprehensive},
overlap detection and the construction of overlap graphs~\cite{minimap2,berlin2015assembling,song2020overlap}, 
phylogenetic tree reconstruction,
and isoform detection from circular consensus sequence (CCS) reads~\cite{isocon}, to name a few.
The naive all-vs-all comparison gives the most comprehensive information
but does not scale well.
An efficient and widely-used approach that avoids unnecessary comparisons is \emph{bucketing}:
a linear scan is employed to assign each sequence into one or multiple buckets,
followed by pairwise comparisons within each bucket.
The procedure of assigning sequences into buckets, which we refer to as a \emph{bucketing function},
is desired to be both ``sensitive'', i.e., two similar sequences ideally appear
in at least one shared bucket so that they can be compared,
and ``specific'', i.e., two dissimilar sequences ideally appear in disjoint buckets so that
they can be exempt from comparison.
The criteria of similar/dissimilar sequences are application-dependent;
in this work we study bucketing functions for the edit distance~(Levenshtein distance).


A simple yet popular bucketing function is to put a sequence into buckets labeled with its own $k$-mers.
The popular seed-and-extend strategy~\cite{altschul1990basic,altschul1997gapped} implicitly uses this approach.
Various sketching methods such as
minimizer~\cite{roberts2004reducing,schleimer2003winnowing,roberts2004preprocessor, marccais2018asymptotically}
and universal hitting set~\cite{OPM+17, DGKM19}
reduce the number of buckets a sequence is assigned to by
only considering a subset of representative $k$-mers.
These bucketing methods based on exact $k$-mer matching
enjoyed tremendous success in analyzing next-generation sequencing~(NGS) data,
but are challenged by the third-generation long-reads sequencing data
represented by PacBio~\cite{rhoads2015pacbio} and Oxford Nanopore~\cite{jain2018nanopore} technologies;
due to the high error rate, sequences that should be assigned to the same
buckets hardly share any identical $k$-mers~(for a reasonably large $k$ such as
$k = 21$ with 15\% error rate), and therefore results in poor sensitivity.


To address this issue, it is required to be able to recognize similar but not necessarily identical sequences.
A general solution is locality-sensitive hashing (LSH)~\cite{omh, Mc21} where
with high probability, similar sequences are sent
into the same bucket~(i.e., there is a hash collision), and with high
probability dissimilar sequences are sent into different buckets.
However, designing locality-sensitive hashing functions for the edit distance
is hard; the state-of-the-art method Order Min Hash (OMH) is
proved to be a gapped LSH but admits a large gap~\cite{omh}.
Another related approach is embedding the metric space induced by the edit distance
into more well-studied normed spaces~\cite{BJKK04, OR07, song2020overlap}.
However, such an embedding is also hard;
for example, it is known that the embedding into $L_1$ cannot be distortion-free~\cite{KR09}.
In addition, there are seeding/sketching methods such as spaced $k$-mer~\cite{flash, patternhunter},
indel seeds~\cite{mak2006indel}, and the more recent strobemer~\cite{strobemer}
that allow gaps in the extracted seeds to accommodate some edits,
but an edit that happens within the chosen seed can still cause mismatches.

It is worth noting that locality-sensitive hashing functions, when interpreted as bucketing functions,
assign a sequence into exactly one bucket: buckets are labeled with hash values,
and a sequence is put into the single bucket where it is hashed to.
In this work, we propose the concept of \emph{locality-sensitive bucketing}~(LSB) functions
as a generalization of LSH functions by allowing it to assign a sequence into multiple buckets.
Formally, a bucketing function, which maps a sequence (of fixed length) into 
one or more buckets, is defined to be $(d_1, d_2)$-sensitive if any two sequences
within an edit distance of $d_1$ are mapped into at least one shared bucket,
and any two sequences with an edit distance at least $d_2$
are mapped into disjoint subsets of buckets.
While a stochastic definition by introducing a distribution on a family of bucketing
functions can be made in a similar way as the definition of LSH functions,
here we focus on this basic, deterministic definition.
We design several LSB functions for a variety of values of $(d_1,d_2)$
including both ungapped~($d_2 = d_1 + 1$) and gapped~($d_2 > d_1 + 1$) ones.
This demonstrates that allowing one sequence to appear in multiple buckets
makes the locality-sensitive properties easier to satisfy.
Moreover, our lower bound proof shows that any $(1,2)$-sensitive bucketing function
must put each sequence~(of length $n$) into at least $n$ buckets
(see Lemma~\ref{lem:int-1-fs-minimum}),
suggesting that certain ungapped locality-sensitive hashing functions,
where each sequence is sent to a single bucket,
may not exist.

The rest of this paper is organized as follows.
In Section~\ref{sec:prelim}, we give the precise definition of LSB functions
and propose criteria to measure them.
In Sections~\ref{sec:integers} and \ref{sec:sampling}, we design LSB functions using two different approaches,
the results are summarized in Section~\ref{sec:summary}.
We show experimental studies in Section~\ref{sec:experiments}, with a focus on demonstrating
the performance of gapped LSB functions.
Future directions are discussed in Section~\ref{sec:discussion}.

\section{Basics of locality-sensitive bucketing~(LSB) functions}\label{sec:prelim}
Given an alphabet $\Sigma$ with $|\Sigma|>1$ and a natural number $n$,
let $\S_n=\left(\Sigma^n, \textrm{edit}\right)$ be the metric space
of all length-$n$ sequences equipped with the Levenshtein (edit) distance.
Given a set $B$ of buckets, a bucketing function $f$ maps $\S_n$
to $\P(B)$, the power set of $B$.
This can be viewed as assigning a sequence $\kmer{s}$ of length $n$ to a subset
of buckets $f(\kmer{s})\subset B$.
Let $d_1 < d_2$ be two non-negative integers,
we say a bucketing function $f$ is \emph{$\left(d_1, d_2\right)$-sensitive} if
\begin{align}
  \edit{s}{t} \leq d_1 \implies f(\kmer{s})\cap f(\kmer{t})\neq\emptyset,
  \label{eq:lsb1}\\
  \edit{s}{t} \geq d_2 \implies f(\kmer{s})\cap f(\kmer{t})=\emptyset.
  \label{eq:lsb2}
\end{align}

We refer to the above two conditions as LSB-properties~\eqref{eq:lsb1} and
\eqref{eq:lsb2} respectively.
Intuitively, the LSB-properties state that, if two length-$n$ sequences
are within an edit distance of $d_1$, then the bucketing function
$f$ guarantees assigning them to at least one same bucket, and
if two length-$n$ sequences have an edit distance
at least $d_2$, then the bucketing function $f$ guarantees not assigning
them to any shared bucket.
In other words, $(d_1,d_2)$-sensitive bucketing functions perfectly distinguish length-$n$ sequences
within distance $d_1$ from those with distances at least $d_2$.
It is easy to show that if $f:\S_n\to\P(B)$ is a $(d_1, d_2)$-sensitive bucketing function, then
$f(\kmer{s})\neq\emptyset$ for all $\kmer{s}\in\S_n$.
In fact, since $\edit{s}{s}=0\leq d_1$,
the LSB-property~\eqref{eq:lsb1} implies that $f(\kmer{s})=f(\kmer{s})\cap f(\kmer{s})\neq \emptyset$.
If $d_1=d_2-1$ then we say the bucketing function is ungapped;
otherwise it is called gapped.

We note that the above definition of LSB functions generalize the (deterministic)
LSH functions: if we require that $|f(\kmer{s})| = 1$
for every sequence $\kmer{s}\in \S_n$, i.e., $f$ maps a sequence to a single bucket,
then $f(\kmer{s})\cap f(\kmer{t})\neq\emptyset$ implies $f(\kmer{s}) = f(\kmer{t})$ 
and $f(\kmer{s})\cap f(\kmer{t}) = \emptyset$ implies $f(\kmer{s})\neq f(\kmer{t})$.

Two related parameters can be used to measure an LSB function:
$|B|$, the total number of buckets, and $|f(\kmer{s})|$,
the number of different buckets that contain a specific sequence $\kmer{s}$.
From a practical perspective, it is desirable to keep both parameters small. 
We therefore aim to design LSB functions that minimize $|B|$ and $|f(\kmer{s})|$.
Specifically, in the following sections, we will construct $(d_1, d_2)$-sensitive 
bucketing functions with a variety of values of $(d_1,d_2)$,
and analyze their corresponding
$|B|$ and $|f(\kmer{s})|$; we will also prove lower bounds of 
$|B|$ and $|f(\kmer{s})|$ in different settings and show that some of our constructed
LSB functions are optimal, in terms of minimizing
these two parameters.

The bounds of $|B|$ and $|f(\kmer{s})|$
are closely related to the structure of the metric space $\S_n$.
For a sequence $\kmer{s}\in\S_n$, its $d$-neighborhood,
denoted by $N_n^d(\kmer{s})$, is the subspace of all sequences of length $n$
with edit distance at most $d$ from $\kmer{s}$; 
formally $N_n^d(\kmer{s}) = \{\kmer{t} \in \S_n \mid \textrm{edit}(\kmer{s}, \kmer{t}) \le d\}$.
The following simple fact demonstrates
the connection between the bound of
$|f(\kmer{s})|$ and the structure
of $\S_n$, which will be used later.


\begin{Lemma}\label{lem:ind-set2}
  Let $\kmer{s}$ be a sequence of length $n$.
  If $N_n^{d_1}(\kmer{s})$ contains a subset $X$ with $|X|=x$
  such that every two sequences in $X$ have an edit distance at least $d_2$,
  then for any $(d_1,d_2)$-sensitive bucketing function $f$ we must
  have $|f(\kmer{s})| \ge x$.
\end{Lemma}%
\begin{proof}
Let $f$ be an arbitrary $(d_1,d_2)$-sensitive bucketing function.
By the LSB-property~\eqref{eq:lsb2}, these $x$ sequences
must be assigned to distinct buckets by $f$.
On the other hand, since they are all in $N_n^{d_1}(\kmer{s})$,
the LSB-property~\eqref{eq:lsb1} requires that $f(\kmer{s})$ overlaps
with $f(\kmer{t})$ for each sequence $\kmer{t}\in X$.
Combined, we have $|f(\kmer{s})| \ge x$.
\end{proof}

\section{An optimal \texorpdfstring{$(1,2)$}{(1,2)}-sensitive bucketing function}\label{sec:integers}
In the most general setting of LSB functions, the labels of buckets in $B$ are just symbols
that are irrelevant to the construction of the bucketing function.
Hence we can let $B=\{1, \ldots, |B|\}$.
The remaining of this section studies $(1, 2)$-sensitive bucketing functions in
this general case.
We first prove lower bounds of $|B|$ and $|f(\kmer{s})|$ in this setting;
we then give an algorithm to construct an optimal $(1, 2)$-sensitive bucketing function $f$
that matches these bounds.

\begin{Lemma}\label{lem:int-1-fs-minimum}
  If $f: \S_n \to \P(B)$ is $(1,2)$-sensitive, then for each $\kmer{s}\in\S_n$, $|f(\kmer{s})|\geq n$.
\end{Lemma}
\begin{proof}
  According to Lemma~\ref{lem:ind-set2} with $d_1 = 1$ and $d_2 = 2$,
  we only need to show that $N_n^1(\kmer{s})$ contains 
  $n$ different sequences with pairwise edit distances at least $2$.
  For $i=1, \ldots, n$, let $\kmer{t}^i$ be a sequence obtained from $\kmer{s}$
  by a single substitution at position $i$.
  If $i\neq j$, then $\kmer{t}^i$ differs from $\kmer{t}^j$ at two positions, namely $i$ and $j$.
  Then we must have $\textrm{edit}\left(\kmer{t}^i, \kmer{t}^j\right) \ge 2$ as $\kmer{t}^i$ cannot be
  transformed into $\kmer{t}^j$ with a single substitution or a single insertion or deletion.
  Hence, $\left\{\kmer{t}^1,\ldots, \kmer{t}^n\right\}$ forms the required set.
\end{proof}

\begin{Lemma}\label{lem:int-1-minimum}
  If $f: \S_n \to \P(B)$ is $(1,2)$-sensitive, then $|B|\geq n|\Sigma|^{n-1}$.
\end{Lemma}
\begin{proof}
  Consider the collection of pairs
  $H=\left\{(\kmer{s}, b) \,\middle|\, \kmer{s}\in \S_n\text{ and } b\in f(\kmer{s})\right\}$.
  We bound the size of $H$ from above and below.
  For an arbitrary sequence $\kmer{s}$, let $b\in f(\kmer{s})$ be a bucket that
  contains $\kmer{s}$.
  According to the LSB-property~\eqref{eq:lsb2}, any other sequence in $b$ has edit
  distance $1$ from $\kmer{s}$, \ie, a substitution.
  Suppose that the bucket $b$ contains two sequences $\kmer{u}$ and $\kmer{v}$ that are
  obtained from $\kmer{s}$ by a single substitution at different positions.
  Then $\edit{u}{v}=2$ and $f(\kmer{u})\cap f(\kmer{v})\neq\emptyset$,
  which contradicts the LSB-property~\eqref{eq:lsb2}.
  Therefore, all the sequences in $b$ can only differ from $\kmer{s}$
  at some fixed position $i$.
  There are $|\Sigma|$ such sequences (including $\kmer{s}$ itself).
  So each bucket $b\in B$ can appear in at most $|\Sigma|$ pairs in $H$.
  Thus $|H|\leq |\Sigma|\cdot|B|$.
  
  On the other hand, for a length-$n$ sequence $\kmer{s}$,
  its $1$-neighborhood $N^1_n(\kmer{s})$
  contains $n(|\Sigma|-1)$ other length-$n$ sequences,
  corresponding to the $|\Sigma|-1$ possible
  substitutions at each of the $n$ positions.
  The LSB-property~\eqref{eq:lsb1} requires that $\kmer{s}$ shares at least one bucket with each of them.
  As argued above, each bucket $b\in f(\kmer{s})$ can contain at most
  $|\Sigma|-1$ sequences other than $\kmer{s}$.
  Therefore, $\kmer{s}$ needs to appear in at least
  $n(|\Sigma|-1)/(|\Sigma|-1)=n$ different buckets,
  and hence at least $n$ pairs in $H$.
  So $|H|\geq n |\S_n|=n|\Sigma|^n$.
  Together, we have $|\Sigma|\cdot|B|\geq n |\Sigma|^n$, or
  $|B|\geq n |\Sigma|^{n-1}$.
\end{proof}

We now construct a bucketing function $f:\S_n\to \P(B)$ that is $(1,2)$-sensitive using the algorithm given below.
It has exponential running time with respect to $n$ but primarily serves as a constructive proof that $(1,2)$-sensitive bucketing functions exist.
Assign to the alphabet $\Sigma$ an arbitrary order
$\sigma:\{1, \ldots, |\Sigma|\}\to \Sigma$.
The following algorithm defines the function $f$:


\begin{algorithm}[H]
  \DontPrintSemicolon
  \lForEach{$\kmer{s} \in \S_n$}{$f(\kmer{s})=\emptyset$}
  $m\gets 1$ \quad\tcp{index of the smallest unused bucket}
  \ForEach(\quad\tcp*[h]{in an arbitrary order}){$\kmer{s}=s_1s_2\cdots s_n \in \S_n$}{
    \For{$i=1$ \KwTo $n$}{
      \If(\quad\tcp*[h]{$s_i$ is the smallest character in $\Sigma$}){$s_i==\sigma(1)$}{
        \For{$j=1$ \KwTo $|\Sigma|$}{
          $\kmer{t}\gets s_1\cdots s_{i-1}\sigma(j)s_{i+1}\cdots s_n$\;
          $f(\kmer{t})\gets f(\kmer{t})\cup \{m\}$ \quad\tcp{add $\kmer{t}$ to bucket $m$}
        }
        $m\gets m+1$\;
      }
    }
  }
\end{algorithm}
A toy example of the bucketing function $f$ with $n=2$ and
$\Sigma=\{\sigma(1)=\mathrm{A},\sigma(2)=\mathrm{C},
\sigma(3)=\mathrm{G},\sigma(4)=\mathrm{T}\}$ constructed
using the above algorithm~(where the sequences are processed in the
lexicographical order induced by $\sigma$) is given below, followed by
the contained sequences in the resulting buckets.

\begin{table*}[!ht]
  \centering
\begin{tabular}{llll}
  $f(\mathrm{AA})=\{1, 2\}$, &$f(\mathrm{AC})=\{2, 3\}$,
  &$f(\mathrm{AG})=\{2, 4\}$, &$f(\mathrm{AT})=\{2, 5\}$,\\
  $f(\mathrm{CA})=\{1, 6\}$, &$f(\mathrm{CC})=\{3, 6\}$,
  &$f(\mathrm{CG})=\{4, 6\}$, &$f(\mathrm{CT})=\{5, 6\}$,\\
  $f(\mathrm{GA})=\{1, 7\}$, &$f(\mathrm{GC})=\{3, 7\}$,
  &$f(\mathrm{GG})=\{4, 7\}$, &$f(\mathrm{GT})=\{5, 7\}$,\\
  $f(\mathrm{TA})=\{1, 8\}$, &$f(\mathrm{TC})=\{3, 8\}$,
  &$f(\mathrm{TG})=\{4, 8\}$,&$f(\mathrm{TT})=\{5, 8\}$.\\
\end{tabular}
\end{table*}


\begin{table*}[!ht]
  \centering
  \begin{tabular}{ccccc}
    bucket \# & sequences & & bucket \# & sequences\\
    \hline
    1 & AA, CA, GA, TA & &
    2 & AA, AC, AG, AT\\
    3 & AC, CC, GC, TC & &
    4 & AG, CG, GG, TG\\
    5 & AT, CT, GT, TT & &
    6 & CA, CC, CG, CT\\
    7 & GA, GC, GG, GT & &
    8 & TA, TC, TG, TT\\
  \end{tabular}
\end{table*}



\begin{Lemma}\label{lem:int-1-optimal}
  The constructed bucketing function $f:\S_n\to\P(B)$ satisfies:
  (i) each bucket contains $|\Sigma|$ sequences,
  (ii) $|f(\kmer{s})|=n$ for each $\kmer{s}\in\S_n$, and
  (iii) $|B|=n|\Sigma|^{n-1}$.
\end{Lemma}
\begin{proof}
  Claim~(i) follows directly from the construction~(the most inner for-loop).
  In the algorithm, each sequence $\kmer{s} \in \S_n$ is added to $n$ different
  buckets, one for each position.
  Specifically, let $\kmer{s} = s_1s_2\cdots s_n$,
  then $\kmer{s}$ is added to a new bucket when we process the sequence
  $\kmer{s}^i = s_1 s_2 \cdots s_{i-1} \sigma(1) s_{i+1} \cdots s_n$,
  $1 \le i \le n$.
  Hence, $|f(\kmer{s})|=n$.
  To calculate $|B|$, observe that a new bucket is used whenever we encounter the smallest
  character $\sigma(1)$ in some sequence $\kmer{s}$.
  So $|B|$ is the same as the number of occurrences of $\sigma(1)$
  among all sequences in $\S_n$.
  The total number of characters in $\S_n$ is $n|\Sigma|^n$.
  By symmetry, $\sigma(1)$ appears $n|\Sigma|^{n-1}$ times.
\end{proof}

\begin{Lemma}\label{lem:int-1-sensitive}
  The constructed bucketing function $f$ is $(1,2)$-sensitive.
\end{Lemma}
\begin{proof}
  We show that for $\kmer{s}, \kmer{t}\in \S_n$,
  $\edit{s}{t}\leq 1$ if and only if $f(\kmer{s})\cap f(\kmer{t})\neq\emptyset$.
  For the forward direction, $\edit{s}{t} \le 1$ implies that $\kmer{s}$ and $\kmer{t}$
  can differ by at most one substitution at some position $i$.
  Let $\kmer{r}$ be the sequence that is identical to $\kmer{s}$ except
  at the $i$-th position where it is substituted by $\sigma(1)$
  (it is possible that $\kmer{r} = \kmer{s}$).
  According to the algorithm, when processing $\kmer{r}$, both $\kmer{s}$ and $\kmer{t}$
  are added to a same bucket $m$.
  Therefore, $m\in f(\kmer{s})\cap f(\kmer{t})$.
  
  For the backward direction,
  let $m$ be an integer from $f(\kmer{s})\cap f(\kmer{t})$.
  By construction, all the $|\Sigma|$ sequences in the bucket $m$ differ by a single
  substitution.
  Hence, $\edit{s}{t}\leq 1$.
\end{proof}

Combining Lemmas~\ref{lem:int-1-fs-minimum}--\ref{lem:int-1-sensitive},
we have shown that the above $(1,2)$-sensitive bucketing function
is optimal in the sense of minimizing $|B|$ and $|f(\kmer{s})|$. 
This is summarized below.

\begin{theorem}\label{thm:ints-1}
  Let $B=\{1, \ldots, n |\Sigma|^{n-1}\}$, there is a
  $(1, 2)$-sensitive bucketing function $f:\S_n \to \P(B)$ with $|f(\kmer{s})| = n$ for each $\kmer{s} \in \S_n$.
  No $(1,2)$-sensitive bucketing function exists if $|B|$ is smaller or $|f(\kmer{s})|< n$ for some sequence $\kmer{s}\in \S_n$.
\end{theorem}

\section{Mapping to sequences of length \texorpdfstring{$n$}{n}}\label{sec:sampling}
We continue to explore LSB functions with different values of $d_1$ and $d_2$.
Here we focus on a special case where $B\subset \S_n$,
namely, each bucket in $B$ is labeled by a length-$n$ sequence.
The idea of designing such LSB functions is to map a sequence $\kmer{s}$
to its neighboring sequences that are in $B$.
Formally, given a subset $B\subset\S_n$ and an integer $r\geq 1$, 
we define the bucketing function $f^r_B:\S_n\to \P(B)$ by
\[
f^r_B(\kmer{s})= N_n^r(\kmer{s}) \cap B
= \left\{\kmer{v}\in B\,\middle|\, \edit{s}{v}\leq r\right\}
\text{ for each } \kmer{s}\in\S_n.
\]

We now derive the conditions for $f^r_B$ to be an LSB function.
For any sequence $\kmer{s}$, all the buckets in $f^r_B(\kmer{s})$
are labeled by its neighboring sequences within radius $r$.
Therefore, if two sequences $\kmer{s}$ and $\kmer{t}$ share a bucket
labeled by $\kmer{v}$, then $\edit{s}{v}\leq r$ and $\edit{t}{v}\leq r$.
Recall that $\S_n$ is a metric space, in particular, the
triangle inequality holds.
So $\edit{s}{t}\leq \edit{s}{v}+\edit{t}{v}\leq 2r$.
In other words, if $\kmer{s}$ and $\kmer{t}$ 
are $2r + 1$ edits apart, then they will be mapped to disjoint buckets.
Formally, if $\edit{s}{t} \ge 2r + 1$, then $f^r_B(\kmer{s}) \cap f^r_B(\kmer{t}) = \emptyset$.
This implies that $f^r_B$ satisfies the LSB-property~\eqref{eq:lsb2}
with $d_2 = 2r + 1$.
We note that this statement holds regardless of the choice of $B$.

Hence, to make $f^r_B$ a $(d_1, 2r+1)$-sensitive bucketing function for some
integer $d_1$, we only need to determine a subset $B$ so that $f^r_B$
satisfies the LSB-property~\eqref{eq:lsb1}.
Specifically, $B$ should be picked such that 
for any two length-$n$ sequences $\kmer{s}$ and $\kmer{t}$ within an edit distance of $d_1$,
we always have
\[
f_B^r(\kmer{s}) \cap f_B^r(\kmer{t})
= \left(N_n^r(\kmer{s})\cap B\right) \cap \left(N_n^r(\kmer{t}) \cap B\right)
= N_n^r(\kmer{s}) \cap N_n^r(\kmer{t}) \cap B \neq\emptyset.
\]
For the sake of simplicity, we say a set of buckets $B\subset \S_n$ is \emph{$(d_1,r)$-guaranteed}
if and only if $N_n^r(\kmer{s}) \cap N_n^r(\kmer{t}) \cap B \neq\emptyset$ for every pair of sequences $\kmer{s}$ and $\kmer{t}$ with $\edit{s}{t} \le d_1$.
Equivalently, following the above arguments, $B$ is $(d_1,r)$-guaranteed if and only if
the corresponding bucketing function $f_B^r$ is $(d_1, 2r+1)$-sensitive. 
Note that the $(d_1, r)$-guaranteed set is not a new concept,
but rather an abbreviation to avoid repeating the long phrase
``a set whose corresponding bucketing function is $(d_1, 2r+1)$-sensitive''.
In the following sections, we show several $(d_1, r)$-guaranteed
subsets $B\subset\S_n$ for different values of $d_1$.

\subsection{\texorpdfstring{$(2r,r)$}{(2r,r)}-guaranteed and \texorpdfstring{$(2r - 1, r)$}{(2r-1, r)}-guaranteed subsets}

We first consider an extreme case where $B = \S_n$.
\begin{Lemma}\label{lem:int-1-injective}
  Let $B = \S_n$. 
  Then ${B}$ is $(2r,r)$-guaranteed if $r$ is even, 
  and ${B}$ is $(2r-1,r)$-guaranteed if $r$ is odd.  
  \end{Lemma}
\begin{proof}
  First consider the case that $r$ is even.
  Let $\kmer{s}$ and $\kmer{t}$ be two length-$n$ sequences
  with $\edit{s}{t}\leq 2r$.
  Then there are $2r$ edits that transforms $\kmer{s}$ to
  $\kmer{t}$. (If $\edit{s}{t}<2r$, we can add in trivial edits that
  substitute a character with itself.)
  Because $\kmer{s}$ and $\kmer{t}$ have the same length, these $2r$ edits
  must contain the same number of insertions and deletions.
  Reorder the edits so that each insertion is followed immediately by
  a deletion (\ie, a pair of indels) and all the indels come
  before substitutions.
  Because $r$ is even, in this new order, the first $r$ edits contain an
  equal number of insertions and deletions.
  Namely, applying the first $r$ edits on $\kmer{s}$ produces
  a length-$n$ sequence $\kmer{v}$.
  Clearly, $\edit{s}{v}\leq r$ and $\edit{t}{v}\leq r$,
  \ie, $\kmer{v}\in N_n^r(\kmer{s})\cap N_n^r(\kmer{t})=
  N_n^r(\kmer{s})\cap N_n^r(\kmer{t})\cap B$.

  For the case that $r$ is odd.
  Let $\kmer{s}$ and $\kmer{t}$ be two length-$n$ sequences
  with $\edit{s}{t}\leq 2r-1$.
  By the same argument as above, $\kmer{s}$ can be transformed to $\kmer{t}$
  by $2r-1$ edits and we can assume that all the indels appear
  in pairs and they come before all the substitutions.
  Because $r$ is odd, $r-1$ is even.
  So applying the first $r-1$ edits on $\kmer{s}$ produces a length-$n$
  sequence $\kmer{v}$ such that $\edit{s}{v}\leq r-1<r$ and $\edit{t}{v}\leq
  2r-1-(r-1)=r$.
  Therefore, $\kmer{v}\in N_n^r(\kmer{s})\cap N_n^r(\kmer{t})=
  N_n^r(\kmer{s})\cap N_n^r(\kmer{t})\cap B$.
\end{proof}


By definition, setting $B=\S_n$ makes $f^r_B$ $(2r, 2r+1)$-sensitive if
$r$ is even and $(2r-1, 2r+1)$-sensitive if $r$ is odd.
This provides nearly optimal bucketing performance
in the sense that there is no gap~(when $r$ is even) or the gap is just one~(when $r$ is odd).
It is evident from the proof that the gap at $2r$ indeed exists when $r$ is odd
because if $\kmer{s}$ can only be transformed to $\kmer{t}$ by
$r$ pairs of indels, then there is no length-$n$ sequence $\kmer{v}$ with
$\edit{s}{v}=\edit{t}{v}=r$.

\subsection{Properties of \texorpdfstring{$(r,r)$}{(r,r)}-guaranteed subsets}\label{sec:r-guarantee}
In the above section all sequences in $\S_n$ are used as buckets.
A natural question is, can we use a proper subset of $\S_n$
to achieve (gapped) LSB functions?
This can be viewed as down-sampling $\S_n$ such that if two length-$n$
sequences $\kmer{s}$ and $\kmer{t}$ are similar,
then a length-$n$ sequence is always sampled from
their common neighborhood $N_n^r(\kmer{s})\cap N_n^r(\kmer{t})$. 

Here we focus on the case that $d_1 = r$, \ie,
we aim to construct $B$ that is $(r,r)$-guaranteed. 
Recall that this means for any $\kmer{s}, \kmer{t}\in \S_n$ with
$\edit{s}{t}\leq r$, we have $N_n^r(\kmer{s})\cap N_n^r(\kmer{t})\cap B\neq \emptyset$.
In other words, $f^r_B$ is $(r, 2r+1)$-sensitive.
To prepare the construction,
we first investigate some structural properties of $(r,r)$-guaranteed subsets.
We propose a conjecture that such sets form a
hierarchical structure with decreasing $r$:
\begin{Conjecture}
  If $B\subset \S_n$ is $(r,r)$-guaranteed, then $B$ is also $(r+1,r+1)$-guaranteed.
\end{Conjecture}

We prove a weaker statement:
\begin{Lemma}\label{lem:r-implies-rplus2}
  If $B\subset \S_n$ is $(r,r)$-guaranteed, then $B$ is $(r+2,r+2)$-guaranteed.
\end{Lemma}
\begin{proof}
  Let $\kmer{s}$ and $\kmer{t}$ be two length-$n$ sequences
  with $\edit{s}{t} \le r+2$;
  we want to show that $N_n^{r+2}(\kmer{s}) \cap N_n^{r+2}(\kmer{t}) \cap B \neq \emptyset$.
  Consider a list of edits that transforms
  $\kmer{s}$ to $\kmer{t}$: skipping a pair of indels or two
  substitutions gives a length-$n$ sequence $\kmer{m}$
  such that $\edit{s}{m}\le r$ and
  $\edit{t}{m}=2$.   
  Because $\kmer{s}$ and $\kmer{m}$ are within a distance of $r$ and $B$ is $(r,r)$-guaranteed, we have that $N_n^r(\kmer{s}) \cap N_n^r(\kmer{m}) \cap B \neq \emptyset$, i.e., 
  there exists a length-$n$ sequence $\kmer{v}\in B$
  such that $\edit{s}{v}\leq r$ and $\edit{m}{v}\leq r$.
  By triangle inequality,
  $\edit{t}{v}\leq\edit{t}{m}+\edit{m}{v}\leq r+2$.
  Hence, we have $\kmer{v}\in N_n^{r+2}(\kmer{t})$.
  Clearly, $\kmer{v}\in N_n^{r}(\kmer{s})$ implies that
  $\kmer{v}\in N_n^{r+2}(\kmer{s})$.
  Combined, we have $\kmer{v}\in N_n^{r+2}(\kmer{s})\cap N_n^{r+2}(\kmer{t})\cap B$.
\end{proof}

The next lemma shows that $(1,1)$-guaranteed subsets have the strongest condition.
\begin{Lemma}\label{lem:1-implies-all}
  If $B\subset S_n$ is $(1,1)$-guaranteed, then
  $B$ is $(r,r)$-guaranteed for all $r\geq 1$.
\end{Lemma}
\begin{proof}
  According to the previous lemma, we only need to show that
  $B$ is $(2,2)$-guaranteed.
  Given two length-$n$ sequences $\kmer{s}$ and $\kmer{t}$
  with $\edit{s}{t}=2$,
  consider a list $Q$ of two edits that transforms $\kmer{s}$
  to $\kmer{t}$.
  There are two possibilities:
  \begin{itemize}
  \item
    If both edits in $Q$ are substitutions,
    let $i$ be the position of the first substitution.   
  \item
    If $Q$ consists of one insertion and one deletion,
    let $i$ be the position of the character that is going to be
    deleted from $\kmer{s}$.
  \end{itemize}
  In either case, let $\kmer{m}$ be a length-$n$ sequence
  obtained by replacing the
  $i$-th character of $\kmer{s}$ with another character in $\Sigma$.
  Then $\edit{s}{m}=1$.
  Because $B$ is $(1,1)$-guaranteed, there is a length-$n$ sequence
  $\kmer{v}\in B$
  such that $\edit{s}{v}\leq 1$ and $\edit{m}{v}\leq 1$.
  Observe that either $\kmer{s}=\kmer{v}$ or $\kmer{v}$ is obtained from
  $\kmer{s}$ by one substitution at position $i$.
  So applying the two edits in $Q$ on $\kmer{v}$ also produces $\kmer{t}$,
  \ie, $\edit{t}{v}\leq 2$.
  Therefore, $\kmer{v}\in N_n^{2}(\kmer{s})\cap N_n^{2}(\kmer{t})\cap B$.
\end{proof}


Now we bound the size of a $(1,1)$-guaranteed subset from below.
\begin{Lemma}\label{lem:tau1}
  If $B$ is (1,1)-guaranteed, then
  \[
  \text{(i) for each } \kmer{s}\in\S_n,\; \left|N_n^1(\kmer{s})\cap B\right|\geq
  \begin{cases}
    1 & \text{if }\kmer{s}\in B\\
    n & \text{if }\kmer{s}\not\in B
  \end{cases},
  \qquad
  \text{(ii) } |B| \geq |\S_n|/|\Sigma| = |\Sigma|^{n-1}.
  \]
\end{Lemma}
\begin{proof}
  Let $B\subset \S_n$ be an arbitrary $(1,1)$-guaranteed subset.
  For part~(i), because $\kmer{s}\in N_n^1(\kmer{s})$, if $\kmer{s}$ is also
  in $B$, then $\kmer{s}$ is in their intersection, hence
  $\left|N_n^1(\kmer{s})\cap B\right|\geq 1$.
  If $\kmer{s}=s_1s_2\ldots s_n \not\in B$, then it must have
  at least $n$ $1$-neighbors $\kmer{v}^i\in B$,
  one for each position $1\leq i\leq n$, where
  $\kmer{v}^i = s_1\ldots s_{i-1}v_is_{i+1}\ldots s_n$, $v_i\neq s_i$.
  Suppose conversely that this is not the case for a particular $i$.
  Let $\kmer{t}= s_1\ldots s_{i-1}t_is_{i+1}\ldots s_n$ where
  $t_i\neq s_i$. We have $\edit{s}{t}=1$.
  Also, $N_n^1(\kmer{s})\cap N_n^1(\kmer{t}) = \{x\in \Sigma \mid s_1\ldots s_{i-1}xs_{i+1}\ldots s_n\}$,
  but none of them is in $B$~(consider the two cases $x=s_i$ and $x\neq s_i$), i.e., $N_n^1(\kmer{s})\cap N_n^1(\kmer{t}) \cap B = \emptyset$.
  This contradicts the assumption that $B$ is $(1,1)$-guaranteed.
  
  For part~(ii), consider the collection of pairs
  $H=\left\{(\kmer{s}, \kmer{v})\,\middle|\, \kmer{s}\in\S_n \text{ and }
  \kmer{v}\in N_n^1(\kmer{s})\cap B \right\}$.
  For all $\kmer{v}\in B$, the number of sequences $\kmer{s}\in \S_n$ with
  $\edit{s}{v}\leq 1$ is $n\left(|\Sigma|-1\right)+1$.
  So $|H| = \left(n\left(|\Sigma|-1\right)+1\right)|B|$.
  On the other hand, part~(i) implies that
  $|H|\geq |B|+ n\left(|\Sigma|^n-|B|\right)$.
  Combined, we have $|B|\geq |\Sigma|^{n-1}$, as claimed.
\end{proof}

In Section~\ref{sec:1-construction}, we give an algorithm to 
construct a $(1,1)$-guaranteed subset $B$ that achieves
the size $|B| = |\Sigma|^{n-1}$;
furthermore, the corresponding $(1,3)$-sensitive bucketing function $f^1_B$
satisfies
$\left|f^1_B(\kmer{s})\right|=1$ if $\kmer{s}\in B$ and
$\left|f^1_B(\kmer{s})\right|=n$ if $\kmer{s}\not\in B$.
This shows that the lower bounds proved above in Lemma~\ref{lem:tau1} are tight
and that the constructed $(1,1)$-guaranteed subset $B$ is optimal in the sense of minimizing both $|B|$ and $\left|f^1_B(\kmer{s})\right|$.
Notice that this result improves Lemma~\ref{lem:int-1-injective} with $r = 1$ where we showed that $\S_n$ is a $(1,1)$-guaranteed subset of size $|\Sigma|^n$.
According to Lemma~\ref{lem:1-implies-all}, this constructed $B$ is also $(r,r)$-guaranteed.
So the corresponding bucketing function
$f_B^r$ is $(r, 2r+1)$-sensitive for all integers $r\geq 1$.

\subsection{Construction of optimal \texorpdfstring{$(1,1)$}{(1,1)}-guaranteed subsets} \label{sec:1-construction}

Let $m=|\Sigma|$ and denote the characters in $\Sigma$ by
$c_1, c_2, \ldots, c_m$.
We describe a recursive procedure to construct a
$(1,1)$-guaranteed subset of $\S_n$.
In fact, we show that $\S_n$ can be partitioned into $m$
subsets $B_n^1 \sqcup B_n^2\sqcup \cdots \sqcup B_n^m$ such that
each $B_n^i$ is $(1,1)$-guaranteed. 
Here the notation $\sqcup$ denotes disjoint union.
The partition of $\S_{n}$ is built from the partition of $\S_{n-1}$. 
The base case is $\S_1=\{c_1\}\sqcup\cdots\sqcup \{c_m\}$.

Suppose that we already have the partition for
$\S_{n-1} = B_{n-1}^1 \sqcup B_{n-1}^2\sqcup\cdots\sqcup B_{n-1}^m$.
Let
\[
B_n^1=\left(c_1\circ B_{n-1}^1\right) \sqcup
\left(c_2\circ B_{n-1}^2\right) \sqcup\cdots\sqcup
\left(c_m\circ B_{n-1}^{m}\right),
\]
where $c\circ B$ is the set obtained by prepending the character $c$
to each sequence in the set $B$.
For $B_n^2$, the construction is similar where the partitions
of $\S_{n-1}$ are shifted (rotated) by one such that
$c_1$ is paired with $B_{n-1}^2$,
$c_2$ is paired with $B_{n-1}^3$, and so on.
In general, for $1\leq i\leq m$,
\[
B_n^i=\left(c_1\circ B_{n-1}^i\right) \sqcup
\left(c_2\circ B_{n-1}^{i+1}\right) \sqcup\cdots\sqcup
\left(c_{m-i+1} \circ B_{n-1}^m\right) \sqcup
\left(c_{m-i+2} \circ B_{n-1}^1\right) \sqcup\cdots\sqcup
\left(c_m\circ B_{n-1}^{i-1}\right).
\]
Examples of this partition for $\Sigma=\{$A, C, G, T$\}$
and $n=2, 3$ are shown below.
{\small
\begin{align*}
  B_2^1=\{&\textrm{AA, CC, GG, TT}\}\\
  B_2^2=\{&\textrm{AC, CG, GT, TA}\}\\
  B_2^3=\{&\textrm{AG, CT, GA, TC}\}\\
  B_2^4=\{&\textrm{AT, CA, GC, TG}\}\\\\
  B_3^1=\{&\textrm{AAA, ACC, AGG, ATT, CAC, CCG, CGT, CTA,}\\
  &\textrm{GAG, GCT, GGA, GTC, TAT, TCA, TGC, TTG}\}\\
  B_3^2=\{&\textrm{AAC, ACG, AGT, ATA, CAG, CCT, CGA, CTC,}\\
  &\textrm{GAT, GCA, GGC, GTG, TAA, TCC, TGG, TTT}\}\\
  B_3^3=\{&\textrm{AAG, ACT, AGA, ATC, CAT, CCA, CGC, CTG,}\\
  &\textrm{GAA, GCC, GGG, GTT, TAC, TCG, TGT, TTA}\}\\
  B_3^4=\{&\textrm{AAT, ACA, AGC, ATG, CAA, CCC, CGG, CTT,}\\
  &\textrm{GAC, GCG, GGT, GTA, TAG, TCT, TGA, TTC}\}\\
\end{align*}
}%


Note that each sequence in $\S_n$ appears in exactly one of the
subsets $B_n^i$, justifying the use of the disjoint union
notation.
(The induction proof of this claim has identical structure as
the following proofs of Lemma~\ref{lem:1-construction} and
\ref{lem:1-fs},
so we leave it out for conciseness.)
Now we prove the correctness of this construction.
\begin{Lemma}\label{lem:1-construction}
  Each constructed $B_n^i$ is a minimum $(1,1)$-guaranteed subset of $\S_n$.
\end{Lemma}
\begin{proof}
  By Lemma~\ref{lem:tau1}, we only need to show that each $B_n^i$
  is $(1,1)$-guaranteed and has size $|\Sigma|^{n-1} = m^{n-1}$.
  The proof is by induction on $n$.
  The base case $\S_1=\{c_1\}\sqcup \cdots\sqcup\{c_m\}$ is easy to verify.

  As the induction hypothesis, suppose that
  $\S_{n-1}=\bigsqcup_{j=1}^m B_{n-1}^j$, where each
  $B_{n-1}^j$ is $(1,1)$-guaranteed and has size $m^{n-2}$.
  Consider an arbitrary index $1\leq i\leq m$.
  By construction, we have
  $\left|B_n^i\right|=\sum_{j=1}^m \left|B_{n-1}^j\right| = m^{n-1}$.
  To show that $B_n^i$ is $(1,1)$-guaranteed,
  consider two sequences $\kmer{s}, \kmer{t}\in \S_n$ with $\edit{s}{t}=1$.
  If the single substitution happens on the first character,
  let $\kmer{x}\in \S_{n-1}$ be the common $(n-1)$-suffix of $\kmer{s}$
  and $\kmer{t}$.
  Since $\bigsqcup_{j=1}^m B_{n-1}^j$ is a partition of $\S_{n-1}$,
  $\kmer{x}$ must appear in one of the subsets $B_{n-1}^{\ell}$.
  In $B_n^i$, it is paired with one of the characters $c_k$.
  Let $\kmer{y}=c_k\circ\kmer{x}$, then $\kmer{y}\in B_n^i$.
  Furthermore, $\kmer{s}$ and $\kmer{t}$ can each be transformed to
  $\kmer{y}$ by at most one substitution on the first character.
  Thus, $\kmer{y}\in N_n^1(\kmer{s})\cap N_n^1(\kmer{t})\cap B_n^i$.

  If the single substitution between $\kmer{s}$ and $\kmer{t}$ does not
  happen on the first position, then they share the common first
  character $c_k$.
  In $B_n^i$, $c_k$ is paired with one of the subsets $B_{n-1}^{\ell}$.
  Let $\kmer{s'}$ and $\kmer{t'}$ be $(n-1)$-suffixes of $\kmer{s}$
  and $\kmer{t}$, respectively.
  It is clear that $\edit{s'}{t'}=1$.
  By the induction hypothesis, $B_{n-1}^{\ell}$ is $(1,1)$-guaranteed.
  So there is a sequence $\kmer{x}\in B_{n-1}^{\ell}$ of length $n-1$ such that
  $\edit{s'}{x}\leq 1$ and $\edit{t'}{x}\leq 1$.
  Let $\kmer{y}=c_k\circ\kmer{x}$, then $\kmer{y}\in B_n^i$ by the construction.
  Furthermore, $\edit{s}{y}=\edit{s'}{x}\leq 1$ and
  $\edit{t}{y}=\edit{t'}{x}\leq 1$.
  Thus, $\kmer{y}\in N_n^1(\kmer{s})\cap N_n^1(\kmer{t})\cap B_n^i$.
  Therefore, $B_n^i$ is $(1,1)$-guaranteed.
  Since the index $i$ is arbitrary, this completes the proof.
\end{proof}

It remains to show that for each $\kmer{s}\in \S_n$,
$\left|N_n^1(\kmer{s})\cap B_n^i\right|$ matches the lower bound
in Lemma~\ref{lem:tau1}.
Together with Lemma~\ref{lem:1-construction}, this proves that
each constructed $B_n^i$ yields an optimal $(1,3)$-sensitive
bucketing function in terms of minimizing both the
total number of buckets and the number of buckets each length-$n$
sequence is sent to.

\begin{Lemma}\label{lem:1-fs}
  For $\kmer{s}\in \S_n$, each constructed $B_n^i$ satisfies
  $\left|N_n^1(\kmer{s})\cap B_n^i\right| =
  \begin{cases}
    1 & \text{if } \kmer{s}\in B_n^i\\
    n & \text{if } \kmer{s}\not\in B_n^i
  \end{cases}.$
\end{Lemma}
\begin{proof}
  We proceed by induction on $n$.
  The base case $n=1$ is trivially true because
  $|B_1^i|=1$ and all single-character sequences are
  within one edit of each other.
  Suppose that the claim is true for $n-1$.
  Consider an arbitrary index $i$.
  If $\kmer{s}\in B_n^i$, we show that any other length-$n$
  sequence $\kmer{t}\in B_n^i$
  has edit distance at least $2$ from $\kmer{s}$, namely
  $N_n^1(\kmer{s})\cap B_n^i=\{\kmer{s}\}$.
  Let $\kmer{s'}$ and $\kmer{t'}$ be the $(n-1)$-suffixes of $\kmer{s}$
  and $\kmer{t}$ respectively.
  According to the construction, if $\kmer{s}$ and $\kmer{t}$ have the
  same first character, then $\kmer{s'}$ and $\kmer{t'}$ are in the
  same $B_{n-1}^j$ for some index $j$.
  By the induction hypothesis, $\edit{s'}{t'}\geq 2$ (otherwise
  $\left|N_{n-1}^1\left(\kmer{s'}\right)\cap B_{n-1}^j\right|\geq 2$),
  and therefore $\edit{s}{t}=\edit{s'}{t'}\geq 2$.
  If $\kmer{s}$ and $\kmer{t}$ are different at the first character, then
  $\kmer{s'}$ and $\kmer{t'}$ are not in the same $B_{n-1}^j$, so
  $\kmer{s'}\neq\kmer{t'}$ (recall that $B_{n-1}^j$ and $B_{n-1}^{k}$
  are disjoint if $j\neq k$), namely $\edit{s'}{t'}\geq 1$.
  Together with the necessary substitution at the first character,
  we have $\edit{s}{t}=1+\edit{s'}{t'}\geq 2$.

  If $\kmer{s}\not\in B_n^i$, Lemma~\ref{lem:tau1} and \ref{lem:1-construction}
  guarantee that $\kmer{s}$ has $n$ $1$-neighbors $\kmer{v}^{k}$ in $B_n^i$,
  $k=1,\ldots, n$, where $\kmer{v}^{k}$ is obtained from $\kmer{s}$
  by a single substitution at position $k$.
  Let $\kmer{t}\neq \kmer{s}$ be a $1$-neighbor of $\kmer{s}$.
  Since $\kmer{t}$ can only differ from $\kmer{s}$ by a single substitution
  at some position $\ell$, we know that either $\kmer{t}=\kmer{v}^{\ell}$
  or the edit distance between $\kmer{t}$ and $\kmer{v}^{\ell}$ is $1$.
  In the latter case, $\kmer{t}$ cannot be in $B_n^i$ otherwise
  $\left|N_n^1\left(\kmer{v}^{\ell}\right)\cap B_n^i\right| \geq 2$,
  contradicting the result of the previous paragraph.
  Therefore, $N_n^1(\kmer{s})\cap B_n^i=\left\{\kmer{v}^1,\ldots
  \kmer{v}^{n}\right\}$ which has size $n$.
\end{proof}

We end this section by showing that a membership query can be done in $O(n)$ time
on the $(1,1)$-guaranteed subset $B$ constructed above~(\ie, $B=B_n^i$ for some $i$).
Thanks to its regular structure, the query is performed
without explicit construction of $B$.  
Consequently, the bucketing functions using $B$ can be computed without
computing and storing this subset of size $|\Sigma|^{n-1}$.

Specifically, suppose that we choose $B=B_n^i$ for some fixed $1\leq i\leq m$.
Let $\kmer{s}$ be a given length-$n$ sequence; we want to query if $\kmer{s}$ is in $B$ or not.
This is equivalent to determining whether the index of the partition of $\S_n$
that $\kmer{s}$ falls into is $i$ or not.
Write $\kmer{s}=s_1s_2\ldots s_n$ and let $\kmer{s}'=s_2\ldots s_n$
be the $(n-1)$-suffix of $\kmer{s}$.
Suppose that it has been determined that $\kmer{s}'\in B_{n-1}^j$ for
some index $1\leq j\leq m$,
\ie, the sequence $\kmer{s}'$ of length $n-1$ comes from the $j$-th partition of $\S_{n-1}$.
By construction, the index $\ell$ for which $\kmer{s}\in B_n^{\ell}$ is
uniquely determined by the character $s_1=c_k\in \Sigma$ and the index $j$
according to the formula $\ell = (j+m+1-k)\bmod m$.
The base case $n=1$ is trivially given by the design that $c_p\in B_1^p$ for
all $1\leq p\leq m$.
This easily translates into a linear-time algorithm that scans
the input length-$n$ sequence
$\kmer{s}$ backwards and compute the index $\ell$ such that
$\kmer{s}\in B_n^{\ell}$.
To answer the membership query, we only need to check whether $\ell=i$.
We provide an implementation of both the construction
and the efficient membership query of a $(1,1)$-guaranteed subset at
\url{https://github.com/Shao-Group/lsbucketing}.


\subsection{A \texorpdfstring{$(3,5)$}{(3,5)}-sensitive bucketing function} \label{sec:1-hash-collision}
Let $B\subset \S_n$ be one of the constructed $(1,1)$-guaranteed subsets.
Recall that the resulting bucketing function $f^r_B$ is $(r, 2r + 1)$-sensitive
for all integers $r\geq 1$;
in particular, $f^2_B$ is $(2, 5)$-sensitive. We are able to strengthen this result
by showing that $f^2_B$ is in fact $(3,5)$-sensitive.

\begin{theorem}\label{thm:h2c}
  Let $B\subset \S_n$ be a $(1,1)$-guaranteed subset.
  The bucketing function $f^2_B$ is $(3, 5)$-sensitive.
\end{theorem}
\begin{proof}
  As $f_B^r$ is already proved to be $(2,5)$-sensitive,
  to show it is $(3,5)$-sensitive, we just need to prove that,
  for any two sequences $\kmer{s}, \kmer{t}\in \S_n$ with $\edit{s}{t}=3$,
  $f^2_B(\kmer{s})\cap f^2_B(\kmer{t}) = N^2_n(\kmer{s}) \cap N^2_n(\kmer{t}) \cap B \neq\emptyset$.
  If the three edits are all substitutions, then there are
  length-$n$ sequences
  $\kmer{x}$ and $\kmer{y}$ such that
  $\edit{s}{x}=\edit{x}{y}=\edit{y}{t}=1$.
  Since $B$ is $(1,1)$-guaranteed, there is a length-$n$
  sequence $\kmer{z}\in B$
  with $\edit{x}{z}\leq 1$ and $\edit{y}{z}\leq 1$.
  By triangle inequality, $\edit{s}{z}\leq\edit{s}{x}+\edit{x}{z}\leq 2$;
  $\edit{t}{z}\leq\edit{t}{y}+\edit{y}{z}\leq 2$.
  So $\kmer{z}\in N^2_n(\kmer{s}) \cap N^2_n(\kmer{t}) \cap B $.

  If the three edits are one substitution and a pair of indels,
  then there is a length-$n$ sequence $\kmer{x}$ such that $\edit{s}{x}=1$
  and $\edit{x}{t}=2$ where the two edits between $\kmer{x}$ and $\kmer{t}$
  can only be achieved by one insertion and one deletion.
  Let $i$ be the position in $\kmer{x}$ where the deletion between
  $\kmer{x}$ and $\kmer{t}$ takes place.
  Let $\kmer{y}$ be a length-$n$ sequence obtained from $\kmer{x}$ by a
  substitution at position $i$, so $\edit{x}{y}=1$.
  Since $B$ is $(1,1)$-guaranteed, there is a length-$n$ sequence
  $\kmer{z}\in B$ with $\edit{x}{z}\leq 1$ and $\edit{y}{z}\leq 1$.
  Then $\edit{s}{z}\leq \edit{s}{x} + \edit{x}{z}\leq 2$.
  Observe that $\kmer{x}$ and $\kmer{z}$ differ by at most
  one substitution at position $i$, which will be deleted when
  transforming to $\kmer{t}$.
  So the two edits from $\kmer{x}$ to $\kmer{t}$ can also transform
  $\kmer{z}$ to $\kmer{t}$, namely, $\edit{t}{z}\leq 2$.
  Thus, $\kmer{z}\in N^2_n(\kmer{s}) \cap N^2_n(\kmer{t}) \cap B $.
\end{proof}


\section{Summary of proved LSB functions}\label{sec:summary}
We proposed two sets of LSB functions and studied the efficiency of them in terms of $|B|$, the total number
of buckets, and $|f(\kmer{s})|$, the number of buckets a specific
length-$n$ sequence $\kmer{s}$ occupies.
The results are summarized in Table~\ref{tab:results}.
\begin{table}[!ht]
  \centering
  \caption{Results on $(d_1, d_2)$-sensitive bucketing functions
    of length-$n$ sequences.
    Entries with $\leq$ show the best known upper bounds.
    Entries marked with a single star cannot be reduced under the
    specific bucketing method.
    Entries marked with double stars cannot be reduced in general.
    In column $B$, we use $B_n^i$ to refer to a $(1,1)$-guaranteed subset
    constructed in Section~\ref{sec:1-construction}.}
  \label{tab:results}
  \resizebox{\textwidth}{!}{%
  \begin{tabular}{lllllll}
    $(d_1, d_2)$-sensitive & $B$ & $|B|$ & $|f(\kmer{s})|$ & Ref.\\
    \hline
    $(1, 2)$ & $\{1,\ldots, |B|\}$ & $n|\Sigma|^{n-1}$\dstar & $n$\dstar & Theorem~\ref{thm:ints-1}\\
    $(1, 3)$ & $\S_n$ & $|\Sigma|^n$ & $|N_n^1(\kmer{s})|=(|\Sigma|-1)n+1$ & Lemma~\ref{lem:int-1-injective}\\
    $(1, 3)$ & $B_n^i$ & $|\Sigma|^{n-1}$\star & $\begin{cases}1&\text{if } \kmer{s}\in B\\k&\text{if } \kmer{s}\not\in B\end{cases}$\star & Lemma~\ref{lem:tau1}--\ref{lem:1-fs}\\
    $(3, 5)$ & $B_n^i$ & $|\Sigma|^{n-1}$ & $\leq |N_n^2(\kmer{s})|$ & Theorem~\ref{thm:h2c}\\
    $(r, 2r+1)$, $r>1$ & $B_n^i$ & $|\Sigma|^{n-1}$ & $\leq |N_n^{r}(\kmer{s})|$ & Lemma~\ref{lem:1-implies-all}, \ref{lem:1-construction}\\
    $(2r-1, 2r+1)$, $r\geq 3$ odd & $\S_n$ & $|\Sigma|^n$ & $|N_n^r(\kmer{s})|$ & Lemma~\ref{lem:int-1-injective}\\
    $(2r, 2r+1)$, $r\geq 2$ even& $\S_n$ & $|\Sigma|^n$ & $|N_n^r(\kmer{s})|$ & Lemma~\ref{lem:int-1-injective}\\
    \hline
  \end{tabular}}
\end{table}

\commentblock{
\mingfu{revise above table as we are not including $l<n$ methods;
  write explicitly under column $B$ if the optimal $(1,1)$-guaranteed subset is used}

\ke{The commented out corollary in section 4.2 corresponds to the red row
  in the table. The only difference from the previous row is that it might
  have a smaller |B|. We can remove it if you think it's unnecessary.}
}

\section{Experimental results on the gapped LSB functions}
\label{sec:experiments}
Several gapped LSB functions are introduced in Section~\ref{sec:sampling}.
Now we investigate their behavior at the gap.
We pick 3 LSB functions to experiment, corresponding to the rows 2--4 in Table~1.
For $d=1, 2, \ldots, 6$, we generate $100,000$ random pairs
$(\kmer{s}, \kmer{t})$ of sequences of length $20$
with edit distance $d$.
Each one of the picked LSB functions $f^r_B$ is applied and the
number of pairs that share a bucket under $f^r_B$ is recorded.
The code can be found at
\url{https://github.com/Shao-Group/lsbucketing}.
The results are shown in Figure~\ref{fig:lsb}.
\begin{figure}[!ht]
  \includegraphics[width=\textwidth]{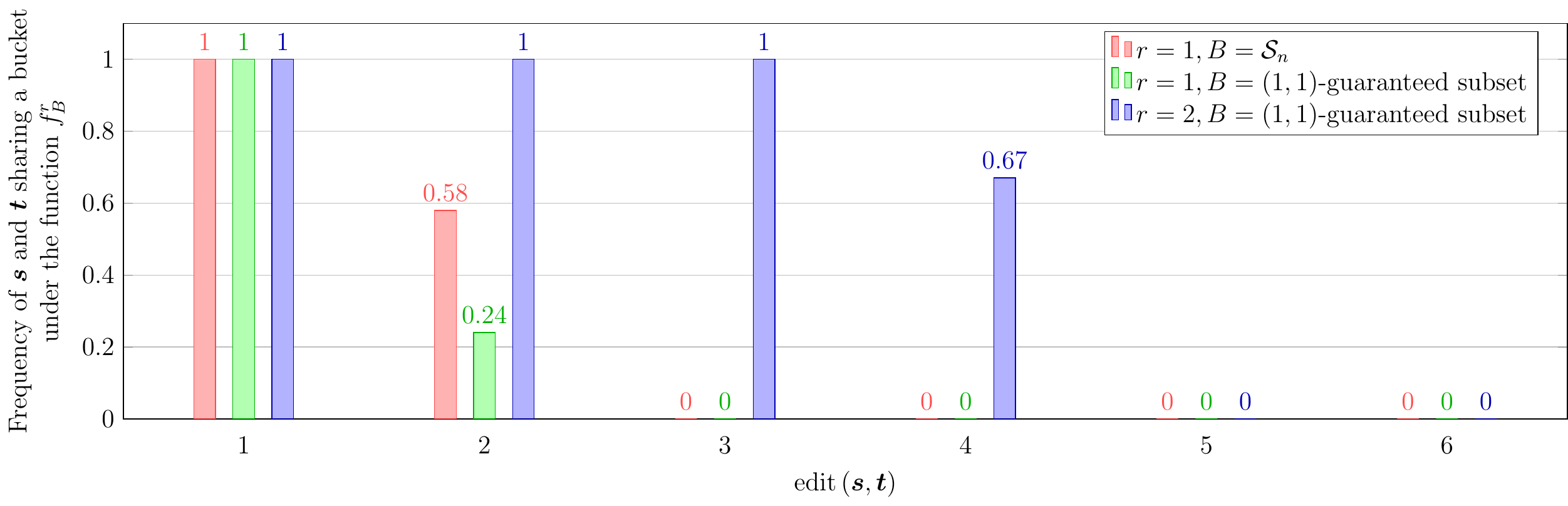}
  \caption{Probabilities (estimated by frequencies) that
    two sequences share a bucket with respect
    to their edit distance under three gapped LSB
    functions~(red, green, and blue bars correspond to
    the rows 2--4 of Table~\ref{tab:results}).
  }\label{fig:lsb}
\end{figure}

Recall that Lemma~\ref{lem:int-1-injective} implies $f^r_{\S_n}$ is
$(2r-1, 2r+1)$-sensitive when $r$ is odd.
The discussion after the proof shows that the gap at $2r$ indeed exists.
In particular, if $\kmer{s}$ can only be transformed to $\kmer{t}$ by
$r$ pairs of indels, then $N_n^r(\kmer{s})\cap N_n^r(\kmer{t})=\emptyset$.
On the other hand, if there are some substitutions among the
$2r$ edits between $\kmer{s}$ and $\kmer{t}$,
then by a similar construction as in the case where
$r$ is even, we can find a length-$n$ sequence $\kmer{v}$ such that
$\edit{s}{v}=\edit{v}{t}=r$.
Motivated by the this observation, we further explore
the performance of the LSB functions at the gap for different types of edits.
Given a gapped LSB function $f$, for the gap at $d$,
define categories $0,\ldots,\lfloor d/2\rfloor$
corresponding to the types of edits:
a pair of length-$n$ sequences with edit distance $d$ is in the $i$-th category
if they can be transformed to each other with 
$i$ pairs of indels (and $d-2i$ substitutions) but not $i-1$ pairs of indels
(and $d-2i+2$ substitutions).
Figure~\ref{fig:gap} shows the results for the three 
LSB functions in Figure~\ref{fig:lsb} at their respective gaps
with respect to different types of edits.
Observe that the result for $f^1_{\S_n}$ (in red) agrees with our analysis above.

\begin{figure}[!ht]
  \includegraphics[width=\textwidth]{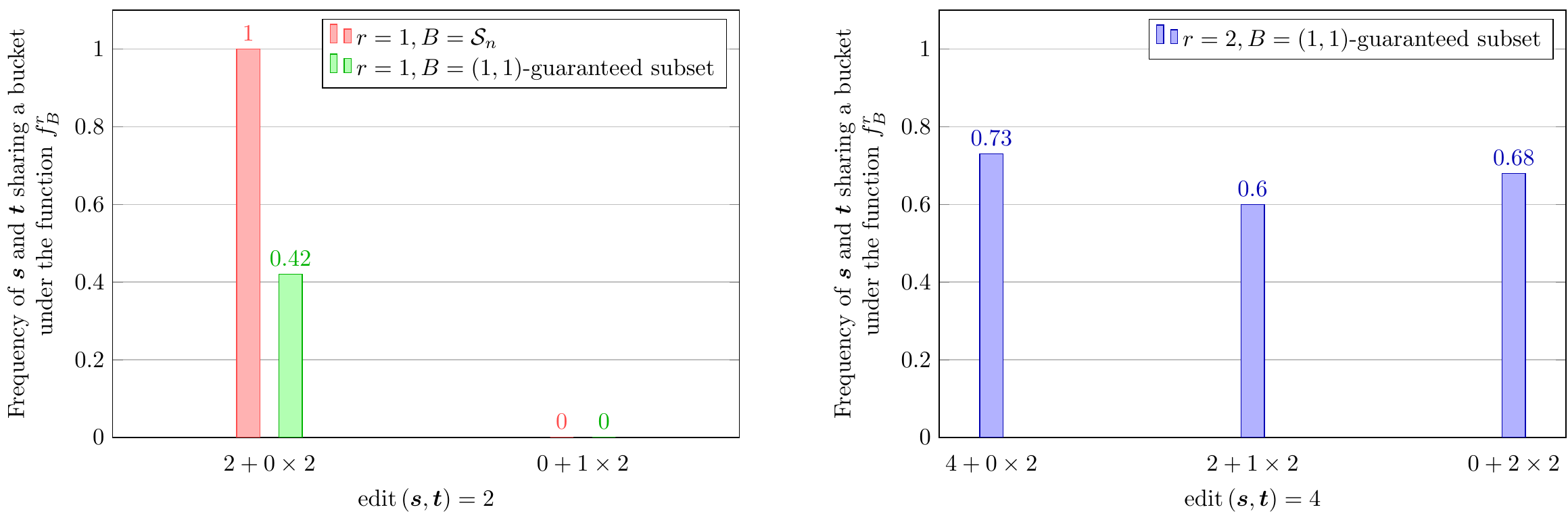}
  \caption{Probabilities (estimated by frequencies) that
    two sequences share a bucket with respect
    to their edit type under three gapped LSB functions.
    The types of edits are labeled in the format $a+b\times 2$ where $a$ is the
    number of substitutions and $b$ is the number of pairs of indels.
    Left: two $(1,3)$-sensitive bucketing functions~(rows 2 and 3 of Table~1).
    Right: the $(3,5)$-sensitive bucketing function~(row 4 of Table~1).
  }\label{fig:gap}
\end{figure}

\section{Conclusion and Discussion} \label{sec:discussion}
\commentblock{
\textcolor{blue}{Move this part to Discussion section?
We can argue that the OR-amplification method is also applicable to our LSB function.}
A commonly used boosting technique known as the OR-amplification combines
multiple locality-sensitive hashing functions in parallel,
which can be viewed as sending each sequence into multiple buckets such that
with high probability, similar sequences share at least one bucket.

- the properties and structures revealed in the proof are of independent interests
for designing LSH functions and beyond.
}

We introduce locality-sensitive bucketing~(LSB) functions, that generalize
locality-sensitive hashing~(LSH) functions by allowing it to map a sequence into
multiple buckets. This generalization makes the LSB functions easier to construct,
while guaranteeing high sensitivity and specificity in a deterministic manner.
We construct such functions, prove their properties, and show that
some of them are optimal under proposed criteria.  
We also reveal several properties and structures of the
metric space $\S_n$, 
which are of independent interests for studying LSH functions and
the edit distance.


Our results for LSB functions can be improved in several aspects.
An obvious open problem is to design $(d_1, d_2)$-sensitive functions
that are not covered here.
For this purpose, one direction is to construct optimal $(r,r)$-guaranteed
subsets for $r>1$.
As an implication of Lemma~\ref{lem:1-fs}, it is worth noting that the
optimal $(1,1)$-guaranteed subset is a maximal independent set in
the undirected graph $G_n^1$ whose vertex set is $\S_n$ and
each sequence is connected to all its $1$-neighbors.
It is natural to suspect that similar results hold for $(r,r)$-guaranteed subsets with larger $r$.
Another approach is to use other more well-studied sets as buckets
and define LSB functions based on their connections with $\S_n$.
This is closely related to the problem of embedding $\S_n$
which is difficult as noted in the introduction.
Our results in Section~\ref{sec:integers} suggest a new angle to
this challenging problem:
instead of restricting our attention to embedding $\S_n$ into metric spaces,
it may be beneficial
to consider a broader category of spaces that are equipped with
a non-transitive relation (here in LSB functions we used subsets of integers with the
``have a nonempty intersection'' relation).
Yet another interesting future research direction would be to
explore the possibility of improving the practical time and space efficiency
of computing and applying LSB functions.

A technique commonly used to boost the sensitivity of an LSH function is known as the OR-amplification.
It combines multiple LSH functions in parallel,
which can be viewed as sending each sequence into multiple buckets such that
the probability of having similar sequences in one bucket is higher
than using the individual functions separately.
However, as a side effect, the OR-amplification hurts specificity: the chance that dissimilar sequences share a bucket also increases.
It is therefore necessary to combine it with other techniques and choosing parameters to balance sensitivity and specificity is a delicate work.
On contrast, the LSB function introduced in this paper
achieves a provably optimal separation of similar and dissimilar sequences.
In addition, the OR-amplification approach can also be applied on top of the
LSB functions as needed.

\bibliography{tolerance}

\end{document}